\documentclass[conference]{IEEEtran}
\IEEEoverridecommandlockouts

\usepackage{amsfonts}
\usepackage[dvips]{graphicx}
\usepackage{times}
\usepackage{cite}
\usepackage{subfigure}
\usepackage{subfig}
\usepackage{caption}
\usepackage{amsmath,amsthm}
\usepackage{array}
\usepackage{amssymb}

\newcommand{\bs}{\boldsymbol}

\usepackage{stfloats}
\usepackage{graphicx}
\usepackage{footnote}
\usepackage{booktabs}
\usepackage{array}
\usepackage{algorithmic}
\usepackage{algorithm}
\usepackage{subeqnarray}
\usepackage{cases}
\usepackage{threeparttable}
\usepackage{color}
\usepackage{url}
\usepackage{caption}

\newtheorem{theorem}{Theorem}

\newtheorem{proposition}{Proposition}
\newtheorem{lemma}{Lemma}

\newtheorem{definition}{Definition}
\newtheorem{assumption}{Assumption}

\newtheorem{corollary}{Corollary}
\newtheorem{mechanism}{Mechanism}

\usepackage{mathtools}

\def\BibTeX{{\rm B\kern-.05em{\sc i\kern-.025em b}\kern-.08em
    T\kern-.1667em\lower.7ex\hbox{E}\kern-.125emX}}
    
\ifodd 0

\newcommand{\com}[1]{\textbf{\color{red} (Comment: #1) }}
\newcommand{\comg}[1]{\textbf{\color{blue} (COMMENT: #1)}}
\newcommand{\response}[1]{\textbf{\color{blue} (RESPONSE: #1)}}
\else

\newcommand{\com}[1]{}
\newcommand{\comg}[1]{}
\newcommand{\response}[1]{}
\fi

\title{Mechanism Design for Large Scale \\Network Utility Maximization}
 \author{Meng Zhang\IEEEauthorrefmark{2} and Deepanshu Vasal\IEEEauthorrefmark{2}\footnote{meng.zhang@northwestern.edu, dvasal@umich.edu}\\ 
 \IEEEauthorblockA{\IEEEauthorrefmark{2}Department of Electrical and Computer Engineering, Northwestern University}\vspace{-12pt}\\
}
\begin{document}

\maketitle
\begin{abstract}
Network utility maximization (NUM) is a general framework for designing distributed optimization algorithms for large-scale networks. An economic challenge arises in the presence of strategic agents' private information.
Existing studies proposed (economic)
\textit{mechanisms} but largely neglected the issue of large-scale  implementation. Specifically, they require certain modifications to the deployed algorithms, which may bring the significant cost.
To tackle this challenge, we present
the large-scale Vickery-Clark-Grove (VCG) Mechanism for NUM, with a simpler payment rule  characterized by the shadow prices.
The Large-Scale VCG Mechanism maximizes the network utility and
achieves individual rationality and budget balance. With infinitely many agents, agents' truthful reports of their types are their dominant strategies; for the finite case, each agent's incentive to misreport converges quadratically to zero.
For practical implementation,
we introduce a modified mechanism that possesses an additional important technical property, \textit{superimposability}, which makes it able to be built upon 
 any (potentially distributed) algorithm that optimally solves the NUM Problem and ensures all agents to obey the algorithm.
We then extend this idea to the dynamic case, when agents' types are dynamically evolving as a controlled Markov process. In this case, the mechanism leads to incentive compatible actions of agent for each time slot.
\end{abstract}
\section{Introduction}
\subsection{Motivations}
Internet is ubiquitous these days with ever-increasing number of smartphones and computers. Network utility maximization (NUM) has been developed as a general framework for scheduling and allocating network resources (e.g., bandwidth, computing power, and storage) to users \cite{NUM,NUM2}. A typical NUM aims to maximize the network utility (the aggregate agents' utility) subject to some network constraints. For instance, the well-known transmission control protocol (TCP) can be interpreted as a distribution solution to some optimization problem in the form of the basic NUM \cite{NUM2}.

The \textit{large scale} of a  networked system makes it necessary to apply a \textit{distributed algorithm} to a NUM problem for efficient network resource allocation. Specifically, since the components of a network are owned by different entities (agents) in general, 
the network designer does not have the complete network information to solve the NUM directly. Even if agents are obedient to share their respective information, gathering such information is impractical  due to the unaffordable communication overheads. Fortunately, the NUM problems in general can be solved in a distributed manner, mainly because of the possibility to decompose it into several subproblems. With such a decomposable structure, one can design an algorithm in which agents solve their individual subproblems coordinated by proper signalings (often in the form of dual variables \cite{NUM}).

Distributed optimization algorithms presume   \textit{obedient} agents  that are willing to follow the algorithms. However,
agents in practice can be strategic and self-interested, i.e., their own objectives differ from the system-level objective. Therefore, an agent
may have incentive to misreport such information so as to manipulate the system outcomes to their own advantages.

There is a large and very complex interaction of people with each other, with the government and the firms. 
Thus it is quite an important problem to design large-scale networks such that when interacted upon by people, who are self-interested agents and try to maximize their own utilities, the outcome is desirable to the planner. 
Game theory is a powerful tool to analyze behaviors among strategic agents. An engineering side of game theory is mechanism design which aims to design systems such that when played on by strategic agents who optimize their individual objectives, they achieve the same objective as envisioned by the designer. 
One of the most widely used examples of mechanism design used in the real-world is auctions~\cite{My81} where an auctioneer asks for bids by the agents. The auction is designed in such a way that when the strategic bidders bid on the value to maximize their own valuations, it maximizes the returns of the auctioneer. Other than the problem of maximizing the utility of the system designer, another common formulations in Mechanism design is to maximize sum of the utilities of all the agents, also called social welfare or network utility \cite{NUM}.

\subsection{Solution Approach}
One approach to preventing such manipulation is \textit{mechanism design}, which aims to design games such that when played by the self-interested agents while the induced game-theoretic equilibrium leads to an efficient allocation. 
 One such celebrated mechanism is the Vickery-Clark-Grove (VCG) Mechanism~\cite{Vickery,Clarke,Groves} that ensures that truthfully revealing each agent's information is a dominant strategy. 
There are two main drawbacks that prevent the VCG Mechanism from being practically implemented in large-scale systems, namely unaffordable communications and computation overheads. In particular, the VCG Mechanism may incur significantly communication overheads since it requires each agent to submit its entire utility function carrying infinitely many messages. Related research efforts (e.g. \cite{Yang,JoTs,Jain,Kakhbod,Sharma,Sinha,GeBerry,Zhang}) have been proposed towards resolving such an issue, in which case users are only required to submit limited messages while the network utility is maximized at the induced equilibria.
The second main issue of the VCG Mechanism that this work considers is the computation issue:
 as the number of agents becomes large, it becomes increasingly difficult to compute the optimal outcomes and implement the VCG mechanism. 

This work aims at resolving the strategic manipulation issue
by extending the idea of strategy-proofness in the large in \cite{large} to the networked systems. Specifically, in a large networked system, each agent regards the ``(shadow) prices'' (the dual variables corresponding to the NUM Problem) as exogenous. Hence, the resulted price-taking behaviors lead to agent willingness to truthfully report. Such a property leads to simple payment rules and hence can relieve the aforementioned drawbacks of the VCG Mechanism.

However, there is another overlooked challenge unsolved, namely the \textit{implementation cost} of mechanisms. Specifically, all existing economic mechanisms proposed in the literature are built upon specific algorithms  or require certain modifications to the existing algorithms \cite{Yang,JoTs,Jain,Kakhbod,Sharma,Sinha,GeBerry,Zhang}). Such a requirement may bring tremendous obstruction for practical implementation: network designers may not be willing merely to ensure incentive compatibility at the cost of modifying algorithms that are already widely adopted (e.g., TCP), which is too costly for large-scale systems. 

As such, this paper also aims at satisfying a new technical property termed \textit{superimposability}. Specifically, we will design a mechanism for large-scale networks that ensures the incentive compatibility of any (potentially distributed) algorithm that maximizes the network utility, by superimposing a payment rule on it; the computation of the payment rule is simple and does not require modifications to the algorithm itself. 
To the best of our knowledge, the only existing superimposable mechanism is the Groves Mechanism \cite{Groves},  but it suffers from severe budget deficits (the network designer needs to inject a large amount of money into the system) and hence is not favorable.

We further extend the idea of the static mechanism (i.e. when agents have a static type and network designer solves a single shot optimization problem) to the dynamic setting when the agents' types are evolving through a controlled Markov process, controlled by their own actions. This is motivated by many practical scenarios where agents repeatedly interact with the central planner while their own preferences of network resources change over time. A similar problem with a finite number of agents was studied in~\cite{BeVa10} where authors extended the idea of the VCG/Pivot Mechanisms to the dynamic setting, and in~\cite{IyJoSu11} where there are infinitely many agents interacting with each other in a bidding mechanism through a mean field. 

\subsection{Contributions}
We summarize our contributions as follows.
\begin{itemize}
    \item \textit{Large-Scale Mechanism Design.} We present the Large-Scale VCG Mechanism which optimally solves a general form of an NUM problem for multi-resource allocation/scheduling.
    It incentivizes agents to truthfully report their types and achieves individual rationality (IR) and budget balance (BB). 

    \item \textit{Convergence Speed.} In a network with a finite number $I$ of agents, we show that each agent's incentive to misreport its type converges to zero at a speed of $\mathcal{O}(1/I^2)$.
    \item \textit{Superimposability.} For practical implementation, we propose a superimposable modification to the Large-Scale VCG Mechanism. That is, itcan be built upon 
 any (potentially distributed) algorithm that optimally solves the NUM Problem and ensures that all agents are willing to obey the algorithm.
    \item \textit{Dynamic Generalizations.} We generalize our mechanism to a dynamic setting
    where types of the players are dynamically evolving through a controlled Markov process.  The dynamic version
    maintains all the properties.
\end{itemize}

The paper is structured as follows. Section \ref{sec:relate} reviews the related studies in the literature.
In Section~\ref{sec:Model} we present the system model and formulate the NUM Problem. In Section~\ref{sec:Mech1}, we present the large-scale mechanism executed in a centralized manner. In Section~\ref{Mech:Impo}, we consider superimposable large-scale mechanism.
In Section~\ref{sec:Mech2}, we consider the case when player's types are dynamically evolving and extend results to such a dynamic setting. We conclude in Section~\ref{sec:Concl}.

 \section{Literature Review}\label{sec:relate}

Extensive studies have been design distributed optimization algorithms for solving NUM Problems without economic concerns (e.g., \cite{NUM,NUM2,VCG3}). For  economic mechanism design (and the corresponding optimization algorithm design),
 the one-shot dominant-strategy allocation mechanism (e.g., the VCG Mechanism) is unfavorable for the NUM problems since it requires significant communications overheads to describe each agent's utility function. 
 Therefore, research efforts on mechanism design for NUM-like problems for allocating divisible network resources, mainly focus on designing Nash mechanisms \cite{Yang,JoTs,Jain,Kakhbod,Sharma,Sinha,Zhang}) (i.e., inducing Nash equilibrium that maximizes the network utility). One exception is \cite{GeBerry}, which proposed a mechanism to quantize each agent's utility function achieving approximately maximization of the network utility. Related studies also considered the issues of
 budget balance (e.g., \cite{VCG1}), combinatorial resource allocations  (e.g., \cite{VCG2}), privacy concerns (e.g., \cite{Kearns}), and dynamic settings (e.g., \cite{IyJoSu11})
 for the VCG Mechanisms in the NUM framework. A recent economic paper \cite{large} studied mechanism design as the number of agent approaches infinity, but without consideration of network constraints. In addition, our mechanism guarantees a faster convergence rate than \cite{large}.

 Another related line of work is the \textit{faithful implementation} of NUM-like algorithms \cite{faith1,faith2,faith3}. These studies designed mechanisms in such a way that  agents are willing to follow certain algorithms that solve the NUM-like problems. Different from the superimposability property considered in our paper, these faithful mechanisms require certain modifications to the original NUM algorithms. For instance, the faithful mechanism in \cite{faith1} requires to add NUM algorithms auxiliary NUM problems, with each problem excluding one specific agent. Hence, network designers may not be willing merely to ensure incentive compatibility at the cost of modifying algorithms that are already widely adopted.
 
\section{System Model and Problem Formulation}
\label{sec:Model}

In this section, we introduce a framework of Network Utility Maximization (NUM). We first describe various components of the model and then present the NUM problem. 
We will then present a well-known VCG Mechanism as a benchmark.

\subsection{System Model}
\subsubsection{System Overview}
We consider a networked system with a system designer (e.g., a network operator),
a set $\mathcal{I}=\{1,...,I\}$ of agents (e.g., mobile users) and a set $\mathcal{N}=\{1\leq n\leq N\}$ of network resources (e.g., bandwidth of different network).

\subsubsection{Agent Modeling}
Each agent $i$ is assigned a two-dimensional 
type $\theta_i\in\mathcal{T}$ and $\zeta_i\in\mathcal{Z}$. 
 In particular, $\theta_i$
characterizes its \textit{willingness to pay} and/or \textit{desire for the resources};
$\zeta_i$ characterizes agent $i$'s \textit{network configuration}. Parameters $\theta_i$ and $\zeta_i$ determines agent's $i$ utility and the constraints in the network to be defined soon. We assume that  $|\mathcal{T}|$ and $|\mathcal{Z}|$ are finite. 
We further define $\bs{\rho}\triangleq\{\rho_{\theta,\zeta}\}$ such that
 $\rho_{\theta,\zeta}$ is the portion of agents having type ($\theta,\zeta$). 
Each agent $i$'s type ($\theta_i,\zeta_i$) is its own private information that it may not be willing to share with others.

\subsubsection{Resource Allocation} 

The system designer decides on allocating (network) resources to each agent.
Let $\bs{x}_i\triangleq\{x_{i,n}\}_{n\in\mathcal{N}}$ denote the resources allocated to agent $i$; $x_{i,n}$ is the amount of resource $n$ allocated to
agent $i$. 
 A type-($\theta_i,\zeta_i$) agent's is assigned a utility function $U(\theta_i,\bs{x}_i)$, which is increasing and strictly concave in $\bs{x}_i$.
 
We consider general  resource capacity constraints given by
\begin{align}
    \sum_{i\in\mathcal{I}}f_{\zeta_i,n}(x_{i,n})\leq C_n,~\forall n\in\mathcal{N},\label{constraint11}
\end{align}
where $f_{\zeta_i,n}(x_{i,n})$ denotes the agent $i$'s \textit{influence} to the system constraint $n$, and $C_n$ represents the capacity of resource $n$. This constraint can capture, for instance, a network resource allocation budget constraint. We assume that $f_{\zeta,n}(x_{i,n})$ is increasing and convex in $x_{i,n}$ and $f_{\zeta,n}(0)=0$ for all $\zeta\in\mathcal{Z}$ and $n\in\mathcal{N}$. We also adopt the following assumption throughout this paper:
\begin{assumption}[Monitorable Influence]\label{Assum1}
After each agent $i$ receives its allocation $\bs{x}_i$, the network designer or some other agent can observe the output values of its influence functions $f_{\zeta_i,n}(x_{i,n})$ for all $n$.
\end{assumption}
For instance, without the need of knowing the exact $\zeta_i$ and $x_{i,n}$,
one can observe some agent $i$'s resource consumption or generated interference via network monitoring measurements and various other key performance indicators \cite{KPI}. Such an assumption is also adopted in \cite{Zhang}.

\subsection{Problem Formulation}
The system designer aims at solving the following NUM Problem:
\begin{subequations}\label{NUM}
\begin{align}
  {\rm NUM:}~~~~\max_{\bs{x}} &\quad \sum_{i\in\mathcal{I}}~U(\theta_i,\bs{x}_i) \\
{\rm s.t.} &\quad \eqref{constraint11}
\end{align}
\end{subequations}
Let $(\bs{x}^o,\bs\lambda^o)$ be the optimal primal-dual solution to \eqref{NUM}.
The NUM Problem being convex and the objective function being strictly concave admits a unique optimal solution $\bs{x}^o\triangleq \{\bs{x}_{i}^o\}_{i\in\mathcal{I}}$. In addition, one can verify that the KKT (Karush–Kuhn–Tucker) conditions \cite{convex} ensure that $\bs{x}_{i}^o=\bs{x}_{j}^o$ if $(\theta_i,\zeta_i)=(\theta_j,\zeta_j)$, i.e., agents with the same type should receive the same amount of resources.
With this observation, we can further transform the NUM Problem by fixing the resource allocation for the same type of agents. Specifically,
Let $\bs{z}=\{\bs{z}_{\theta,\zeta}\}_{\theta\in\mathcal{T},\zeta\in\mathcal{Z}}$ denote the resource allocated to each type-$(\theta,\zeta)$ agent, i.e., $\bs{z}_{\theta_i,\zeta_i}=\bs{x}_i$ for all $i\in\mathcal{I}$.
One can then formulate the Transformed NUM (T-NUM) Problem into the following:
\begin{subequations}\label{T-NUM}
\begin{align}
  {\rm T-NUM:}~~\max_{\bs{z}} &\quad \sum_{\theta\in\mathcal{T}}\sum_{\zeta\in\mathcal{Z}}\rho_{\theta,\zeta}U(\theta,\bs{z}_{\theta,\zeta}) \\
{\rm s.t.} &\quad     \sum_{\theta\in\mathcal{T}}\sum_{\zeta\in\mathcal{Z}}\rho_{\theta,\zeta}f_{\zeta,n}(z_{\theta,\zeta,n})\leq C_n,~\forall n\in\mathcal{N}.
\end{align}
\end{subequations}

The above problem can be readily solved efficiently if the network designer has complete information. However, such an approach is not self-enforcing and each agent may report its type $(\theta_i,\zeta_i)$ to its own advantage.
Hence, solving the NUM Problem requires one to design a proper mechanism to elicit such information and optimally solve the above problem. Specifically, a mechanism aims to achieve the following properties:
\begin{itemize}
    \item (P1) \textit{Efficiency}: The mechanism should induce a Nash equilibrium at which the NUM Problem is optimally solved.
    \item (P2) \textit{Dominant-Strategy Incentive Compatibility (DSIC)}: At the efficient Nash equilibrium, every agent's truthful reporting strategy is a dominant strategy.
    \item (P3) \textit{Individual Rationality (IR)}: Every agent should not be worse off by participating into the mechanism.
    \item (P4) \textit{Budget Balance (BB)}: The aggregate payment from agents to the system designer is non-negative, i.e., the system designer is not required to inject money into the system.
    \item (P5) \textit{Superimposability}: The mechanism can implement any (potentially distributed) algorithm that solves the NUM Problem by superimposing a payment rule on it; the computation of the payment rule is simple and does not require modifications to the algorithm itself. 
\end{itemize}

\subsection{The VCG Mechanism}
In the following, we revisit the VCG mechanism, which is well known to be able to achieve all the above properties (P1)-(P4) \cite{Vickery,Clarke,Groves}:

\begin{mechanism}[VCG Mechanism]
The VCG Mechanism consists of the following two components:
\begin{itemize}
    \item \textbf{Message space:} Each agent $i$ sends a report of its type $r_i\triangleq\left(\tilde{\theta}_i,\tilde{\zeta}_i\right)
    \in\mathcal{T}\times\mathcal{Z}$;
    \item \textbf{Outcome functions:}
 Based on all agents' reports $\bs{r}\triangleq\{r_i\}_{i\in\mathcal{I}}$, 
the system designer selects the allocation, $\boldsymbol{x}^*(\bs{r})$, and the payment from each agent $i$, $h_i(\bs{r})$, such that 
\begin{subequations}
\begin{align}
    \boldsymbol{x}^*(\bs{r})&=\arg\max_{\bs{x}\in\tilde{\mathcal{X}}(\bs{r})}\sum_{i\in\mathcal{I}}U\!(\tilde{\theta}_i,\bs{x}_i),\label{Eq1}\\
    {h}_i(\bs{r})&=\!\!\!\!\max_{\bs{x}_{-i}\in\tilde{\mathcal{X}}_{-i}(\bs{r}_{-i})}\!\sum_{j\neq i}U\!\left(\tilde{\theta}_j,\bs{x}_j\right)\!-\!\sum_{j\neq i}U\!\left(\tilde{\theta}_j,\bs{x}_j^*(\bs{r})\right),\label{Eq2}
\end{align}
\end{subequations}
where
\begin{align}
    \tilde{\mathcal{X}}(\bs{r})\triangleq&\left\{\bs{x}: \sum_{i\in\mathcal{I}} f_{\tilde{\zeta}_i,n}(x_{i,n})\leq C_n,~\forall n\in\mathcal{N}   \right\},\\
    \tilde{\mathcal{X}}_{-i}(\bs{r}_{-i})\triangleq&\left\{\bs{x}_{-i}: \sum_{j\neq i} f_{\tilde{\zeta}_j,n}(x_{j,n})\leq C_n,~\forall n\in\mathcal{N}   \right\},
\end{align}
$\bs{r}_{-i}\triangleq\{r_j\}_{j\neq i}$, and $\bs{x}_{-i}=\{\bs{x}_j\}_{j\neq i}$.

\end{itemize}
\end{mechanism}
To see (P2) is satisfied, note that each agent $i$'s (effective) payoff satisfies
\begin{align}
    P_i(r_i,\bs{r}_{-i})=~&U(\theta_i,\bs{x}_i^*(\bs{r}))-h_i(\bs{r})\nonumber\\
  =~&
\max_{\bs{x}\in\tilde{\mathcal{X}}(\bs{r})} \left[U(\theta_i,\bs{x}_i)+\sum_{j\neq i}U(\tilde{\theta}_j,\bs{x}_j)\right]\nonumber\\
&-\max_{\bs{x}_{-i}\in\tilde{\mathcal{X}}_{-i}(\bs{r}_{-i})}\sum_{j\neq i}U\left(\tilde{\theta}_j,\bs{x}_j\right)\nonumber\\
\geq~& U(\theta_i,\bs{x}_i^*(r_i,\bs{r}_{-i}))+\sum_{j\neq i}U(\tilde{\theta}_j,\bs{x}_j^*(r_i,\bs{r}_{-i}))\nonumber\\
&-\max_{\bs{x}_{-i}\in\tilde{\mathcal{X}}_{-i}(\bs{r}_{-i})}\sum_{j\neq i}U\left(\tilde{\theta}_j,\bs{x}_j\right),~\forall r_i.
\end{align}
Note that,  $\max_{\bs{x}\in\mathcal{X}} \left[U(\theta_i,\bs{x}_i)+\sum_{j\neq i}U(\tilde{\theta}_j,\bs{x}_j)\right]$ is achieved when type-$(\theta_i,\zeta_i)$ agents reveal their truthful type, i.e., $r_i=(\theta_i,\zeta_i)$. Hence, truthful reporting is each type-$(\theta_i,\zeta_i)$ user's optimal strategy, regardless of the decisions of all other types of users.

However, the VCG Mechanism may not be feasible in practice due to the following reasons. First, it involves solving $I+1$ problems, which is hence not impractical when $I$ is large. Second, it is executed in a centralized manner while the NUM algorithms are in general distributed. Hence, implementing such a mechanism may incur the  significant cost of changing the already widely adopt algorithms.
These two issues motivate us to design mechanisms for large-scale networks next.

\section{Large-Scale Mechanism Design}
\label{sec:Mech1}
To resolve the two challenges of the VCG Mechanism above, 
in this section, we first show that the VCG Mechanism converges to a simple form of payment rules. Based on this insight,
we design the (centralized) large-scale VCG Mechanism. We show that the mechanism can achieve (P1)-(P4) when $I$ is infinite and show that (P2) is asymptotically achieved for a finite $I$ with a convergence rate of $\mathcal{O}(1/I^2)$. Based on the mechanism designed in this section, we will further consider a superimposable version in Section \ref{Mech:Impo}.

\subsection{Convergence of the VCG Mechanism}
We start with the following lemma:
\begin{lemma} \label{L1}
The VCG payment for agent $i$ converges to 
\begin{align}
h_i(\bs{r}^*)=\sum_{n\in\mathcal{N}}\lambda_n^o f_{\zeta_i,n}(x_{i,n}^o), \label{h_1}
\end{align}
when $I\rightarrow \infty,$ where
$\bs{\lambda}^o\triangleq \{\lambda_n^o\}_{n\in\mathcal{N}}$ are the optimal dual variables corresponding to the constraint of the NUM Problem in \eqref{NUM}.
\end{lemma}

 We present a proof sketch of Lemma \ref{L1} in Appendix \ref{ProofL1}. 
 

\subsection{The Large-Scale VCG Mechanism}
In this subsection,  we formally introduce the Large-Scale VCG Mechanism based on in Lemma \ref{L1}.
\begin{mechanism}[Large-Scale VCG Mechanism] The Large-Scale VCG Mechanism includes:
\begin{itemize}
    \item Message Space: Each agent $i$ reports its type $r_i\in\mathcal{T}\times\mathcal{Z}$;
    \item Allocation Outcome: Let $\rho_{\theta,\zeta}^*(\bs{r})$ be the portion of the agents reporting the type-$(\theta,\zeta)$. The system designer solves the following problem:
\begin{subequations}\label{M}
\begin{align}
&\max_{\bs{z}}~\sum_{\theta\in\mathcal{T}}\sum_{\zeta\in\mathcal{Z}}~{\rho}^*_{\theta,\zeta}(\bs{r}) U(\theta, \bs{z}_{\theta,\zeta})\\
&~~{\rm s.t.}~\sum_{\theta\in\mathcal{T}}\sum_{\zeta\in\mathcal{Z}}\rho^*_{\theta,\zeta}(\bs{r})f_{\zeta,n}(z_{\theta,\zeta,n})\leq C_n,~\forall n\in\mathcal{N}.\label{constraint}
\end{align}
\end{subequations}
Let $\bs{z}^*(\bs{r})=\{\bs{z}^*_{\theta,\zeta}(\bs{r})\}$ be the optimal primal solution and $\bs{p}^*(\bs{r})=\{p_n^*(\bs{r})\}$ be the optimal dual solution (or shadow prices) to the problem in \eqref{M}. Each agent $i$ receives an allocation 
\begin{align}
    \bs{x}_{i}^*(\bs{r})=\bs{z}_{r_i}^*(\bs{r}),
\end{align}
and needs to pay
\begin{align}
    h_{i}^*(\bs{r})=\sum_{n\in\mathcal{I}}p_n^*(\bs{r})[f_{{\zeta}_i,n}(z_{r_i,n}^*(\bs{r}))-\beta C_n],\label{pay}
\end{align}
for some parameter $\beta\in[0,1]$ to be selected.
\end{itemize}
\end{mechanism}
Note that, due to Assumption \ref{Assum1}, the payment in \eqref{pay} can directly use $f_{{\zeta}_i,n}(z_{r_i,n}^*(\bs{r}))$ instead of $f_{\tilde{\zeta}_i,n}(z_{r_i,n}^*(\bs{r}))$. 

The Large-Scale VCG Mechanism induces a game where each agent $i$ chooses a report, aiming at maximizing its payoff given by
\begin{align}
P_i(\bs{r})\triangleq U(\theta,\bs{z}_{r_i}^*(\bs{r}))-h_i^*(\bs{r}),~\forall i\in\mathcal{I}.
\end{align}

Clearly, when all agents are truthful reporting, \eqref{M} ensures that the  T-NUM Problem in \eqref{T-NUM} is solved. Based on Lemma \ref{L1}, we are ready to introduce the following result:

\begin{proposition}\label{P1}
The Large-Scale VCG Mechanism satisfies DSIC (P2).
\end{proposition}

\begin{proof}
Given all other agents' reports' $\bs{r}_{-i}$, each agent's payoff maximization problem is:
\begin{align}\label{payoff1}
   \max_{r_i\in\mathcal{T}\times\mathcal{Z}}~P_i(r_i,\bs{r}_{-i})=& ~U(\theta_i,\bs{x}_i^*(r_i,\bs{r}_{-i}))\nonumber\\
   &-\sum_{n\in\mathcal{N}}p_n^*(r_i,\bs{r}_{-i})f_{{\zeta}_i,n}(z^*_{r_i,n}(\bs{r})),
\end{align}
where the constant $\beta C_n$ is omitted and $p_n^*(\cdot)$ are determined in \eqref{M}. As $I\rightarrow \infty$, the value of 
$\rho_j^*(\bs{r})$ remains the same even if agent $i$ solely deviates from reporting $r_i=(\theta_i,\zeta_i)$, so does $p_n^*(\bs{r})$. Hence, we can safely remove $r_i$ from $p_n^*(\cdot)$ for all $n\in\mathcal{N}$. 

Note that each agent's payoff in \eqref{payoff1} is bounded by
\begin{align}
P_i(r_i,\bs{r}_{-i})\leq \max_{\bs{x}_i\succeq 0} U(\theta_i, \bs{x}_i)-\sum_{n\in\mathcal{N}}p_n^*(\bs{r}_{-i})f_{{\zeta}_i,n}({x}_{i,n}(\bs{r})),\label{Eq11}
\end{align}
which can be readily shown by contradiction.
In addition, the maximizer $\bs{x}_i^o(\bs{r}_{-i})$ of $\max_{\bs{x}_i\succeq 0} U(\theta_i, \bs{x}_i)-\sum_{n\in\mathcal{N}}p_n^*(r_i,\bs{r}_{-i})f_{{\zeta}_i,n}^*(x_{i,n})$ satisfies, for all $i\in\mathcal{I}$ and $n\in\mathcal{N}$,
\begin{align}
    \frac{\partial U(\theta_i,\bs{x}_i^o(\bs{r}_{-i}))}{\partial x_{i,n}}= p_n^*(r_i,\bs{r}_{-i}) \frac{\partial f_{\zeta_i,n}(x_{i,n}^o(\bs{r}_{-i}))}{\partial x_{i,n}}.\label{Eq2}
\end{align}
By analyzing the KKT conditions of \eqref{M}, $\bs{x}_i^o$ is exactly achieved when agent $i$ reports $r_i=(\theta_i,\zeta_i)$, i.e.,
${\bs{x}_i^o}=\bs{x}_i^*(\theta_i,\zeta_i;\bs{r}_{-i})$,
and, from \eqref{payoff},
\begin{align}
&~P_i(r_i,\bs{r}_{-i})\nonumber\\
&=U(\theta_i,\bs{x}_{i}^o(\bs{r}_{-i}))-\sum_{n\in\mathcal{N}} p_n^*(\bs{r})\left[\sum_{r_i}f_{\zeta_i,t}(x_{i,n}^o(\bs{r}_{-i}))-\beta C_n\right]\nonumber\\
&\leq \max_{\bs{x}_i\succeq \bs{0}} U(\theta_i,\bs{x}_i)-\sum_{n\in\mathcal{N}} p_n^*(\bs{r})\left[\sum_{r_i}f_{\zeta_i,t}(x_{i,n})-\beta C_n\right]. \label{Payoff11}
\end{align}
From \eqref{Eq11} and \eqref{Payoff11}, we see that the truthful report $r_i=(\theta_i,\zeta_i)$ maximizes agent $i$'s payoff. Therefore, the Large-Scale VCG Mechanism achieves (P2).
\end{proof}

We next prove the remaining two properties (P3) and (P4) in the following proposition:

\begin{proposition}\label{P2}
The Large-Scale VCG Mechanism satisfies IR and weakly BB. If $\beta=1$, then the strong budget balance is achieved.
\end{proposition}
\begin{IEEEproof}
The dual variable $\bs{p}^*(\bs{r})$ is component-wisely positive. The total received payment from the network designer is
\begin{align}
    \sum_{i\in\mathcal{I}}h_i^*(\bs{r})&= \sum_{n\in\mathcal{N}} p_n^*(\bs{r})\left[\sum_{r_i}f_{\zeta_i,t}(z^*_{r_i,n}(\bs{r}))-\beta C_n\right]\nonumber\\
    = &\sum_{n\in\mathcal{N}}p_n^*(\bs{r})(1-\beta) C_n\geq 0,
\end{align}
which implies weak budget balance. By setting $\beta=1$, we see that the equality holds.

From \eqref{Payoff11}, each agent's payoff when truthful reporting satisfies 
\begin{align}
    P_i(r_i,\bs{r}_{-i})&=\max_{\bs{x}_i\geq 0} \left[U(\theta_i,\bs{x}_i)-\sum_{n\in\mathcal{I}}\lambda_n^*(f_{{\zeta}_i,n}(x_{i,n})-\beta C_n)\right]\nonumber\\
    &\geq U(
\theta_i,\bs{0})-\sum_{n\in\mathcal{N}}\lambda_n^*\cdot 0+\beta C\nonumber\\&\geq 0,
\end{align}
which shows that IR is achieved.


\end{IEEEproof}

\subsection{Convergence to Efficiency }
In the previous section, we assume $I\rightarrow \infty$.
We next show that the Large-Scale VCG Mechanism leads to an incentive of $\mathcal{O}(1/I^2)$ when an agent misreports.

We need  the following mild assumptions in this subsection:

\begin{assumption}\label{Assum2}
The portion of each type-$(\theta,\zeta)$ agents is positive, i.e., $\rho_{\theta,\zeta}>0$, for all $\theta\in\mathcal{T},\zeta\in\mathcal{Z}$.
\end{assumption}

\begin{assumption}\label{Assum3}
Each type-$(\theta,\zeta)$ agent's influence functions $f_{\zeta,n}({z}_{\theta,\zeta,n})$  satisfy that, for all $n\in\mathcal{N}$,
\begin{align}
    f'_{\zeta,n}({z}_{\theta,\zeta,n})\geq L_f.
\end{align}
for some $L_f>0$, 
for all $\theta\in\mathcal{T}$.
\end{assumption}

\begin{assumption}[Smoothness]\label{Assum4}
Each type-$(\theta,\zeta)$ agent's utility ${U}(\theta,\bs{z})$ is $L_{\theta}$-smooth, i.e., the largest absolute value of the eigenvalue of  $\nabla_{\bs{z}}^2 {U}(\theta,\bs{z})$ is no larger than $L_{\theta}$
for all $\bs{z}$ and $\theta\in\mathcal{T}$.
\end{assumption}


\begin{definition}[$\epsilon$-Dominant Strategy Incentive Compatibility]
There exists an $\epsilon>0$ such that, for all $r_i\in\mathcal{T}\times \mathcal{Z}$ and all $ i\in\mathcal{I},$
\begin{align}
    P_i(\theta_i,\zeta_i;\bs{r}_{-i})+\epsilon
    \geq P_i(r_i;\bs{r}_{-i}),~\forall \bs{r}_{-i}.
\end{align}
\end{definition}
That it, the $\epsilon$-version of DSIC implies that, any reporting strategy $r_i$ should lead to an $\epsilon$ larger payoff than the truthful report does. 
\begin{theorem}\label{T1}
Given a finite number of agents $I$, the Large-Scale VCG Mechanism is
$\epsilon$-incentive compatible for all $\epsilon$ satisfying
\begin{align}
    \epsilon\geq \frac{2}{I^2}\frac{|\mathcal{Z}|\sum_{\theta\in\mathcal{T}}L_{\theta}\sum_{n\in\mathcal{N}}C_n^2}{ L_f^2 \min_r\rho_r^4}.
\end{align}
\end{theorem}

Hence, the convergence speed of each agent's incentive to misreport is $\mathcal{O}(\frac{1}{I^2})$. We present the main proof of  Theorem \ref{T1} in Appendix \ref{ProofT1}, which is  based on the implicit function theorem and the harmonic-mean-arithmetic-mean matrix inequality \cite{matrix}.

\section{Superimposable Mechanism Design for Large-Scale Systems}\label{Mech:Impo}

In the previous section, we designed the large-scale VCG Mechanism which is executed in a centralized manner.
In this section, we design a superimposable version in this section to further achieve (P5) while it maintains the other properties (P1)-(P4). That is, with some payment rule, the new mechanism ensures that all agents are willing to follow any arbitrary NUM algorithm that can optimally solve \eqref{NUM}.

\subsection{Decision Rule}
To describe the agents' communications and computation required behaviors, we first define $\bs\xi\triangleq\{\xi(\theta,\zeta)\}_{\theta\in\mathcal{T},\zeta\in\mathcal{Z}}$ as the \textit{decision rule} for some algorithm. In particular, $\zeta(\theta,\zeta)$ characterizes the type-$(\theta,\zeta)$ agent's communications and computation behaviors in such an algorithm
for solving the NUM Problem in \eqref{NUM}. 

However, an agent may deviate from obeying an algorithm to its own advantage 
which may render the NUM Problem not optimally solved. Let 
 $a_i\in\mathcal{A}_i$ be the action of agent $i$ that describes a (communications and computation) behavior in an algorithm and
$\mathcal{A}_i$ is the space of all possible communications and computation behaviors.
We hence define ($\bs{x}^\star(\bs{a})\triangleq\{\bs{x}_i^\star(\bs{a})\}_{i\in\mathcal{I}}$, $\bs{\lambda}^\star(\bs{a})\triangleq\{\bs{x}_i^\star(\bs{a})\}_{i\in\mathcal{I}}$) as the output primal-dual solutions of the algorithm when agents play $\bs{a}$. As the NUM algorithms should obtain the primal-dual solution $(\bs{x}^o,\bs{\lambda}^o)$ of the problem in \eqref{NUM}, we have
\begin{align}
    \bs{x}^\star(\bs\zeta)=\bs{x}^o,\\
    \bs{\lambda}^\star(\bs\zeta)=\bs{\lambda}^o.
\end{align}
That is, agents obeying the NUM algorithm should lead to an optimal primal-dual solution to the problem in \eqref{NUM}.


\subsection{Superimposable Mechanism}

We introduce the following superimposable version of the Large-Scale Mechanism for an arbitrary NUM algorithm corresponding to a decision rule $\bs\zeta$.

\begin{mechanism}[Superimposable Large-Scale Mechanism] The Superimposable Large-Scale Mechanism includes:
\begin{itemize}
    \item Action Space: Each agent $i$ chooses an action $a_i\in\mathcal{A}_i$;
    \item Allocation Outcome: The allocation received by agent $i$ is  $\bs{x}_i^\star(\bs{a})$ and its payment to the network designer is
    \begin{align}
    h_{i}(\bs{a})&=\sum_{n\in\mathcal{I}}\lambda_n^\star(\bs{a})[f_{{\zeta}_i,n}(x_{i,n}^\star(\bs{a}))-\beta C_n].
    \end{align}
 \end{itemize}
\end{mechanism}
The dual solution $\lambda_n^\star(\bs{a})$ is the output of the algorithm corresponding to agents' action profile $\bs{a}$ and $f_{{\zeta}_i,n}(x_{i,n}^\star(\bs{a}))$ is observable (even without the knowledge of $\zeta_i$ or $\bs{x}_i^\star(\bs{a})$) by Assumption \ref{Assum3}. Therefore, we have

\begin{theorem}\label{T2}
The Superimposable Large-Scale Mechanism satisfies (P1), (P5), and ensures that every agent obeying its decision rule $\zeta$ is a Nash equilibrium. That is,
\begin{align}
P_i(\xi(\theta_i,\zeta_i),\bs{\xi}_{-i})\geq P_i(a_i,\bs{\xi}_{-i}),~\forall a_i\in\mathcal{A}_i, i\in\mathcal{I}, \label{Z1}
\end{align}
as $I\rightarrow \infty$.
\end{theorem}
As the Superimposable Large-Scale Mechanism is applicable to any arbitrary NUM algorithm, it satisfies (P1) and (P5).
The proof of \eqref{Z1} is similar to that of Proposition \ref{P1} and hence is omitted here. We denote that \eqref{Z1} is a weaker property that (P2), as \eqref{Z1} only implies that each agent has the incentive to obey its decision rule $\xi(\theta_i,\zeta_i)$ when other agents are obeying theirs. However, it is a reasonable property that is common in algorithmic mechanism design which ensures the incentive compatibility of specific algorithms \cite{faith1,faith2,faith3}. Based on Proposition \ref{P2}, we are ready to derive the following 

\begin{corollary}
The Superimposable Large-Scale Mechanism satisfies IR and weakly BB. If $\beta=1$, then the strong budget balance is achieved.
\end{corollary}

In a nutshell, the Superimposable Large-Scale Mechanism achieves almost all desirable properties (P1)-(P5).

\section{Generalization to Dynamic Settings}
\label{sec:Mech2}
In this section, we generalize the idea of the large-scale mechanism into dynamic environments. 

\subsection{Dynamic System Model and Problem Formulation}
We consider a slotted time model and the duration of each slot is normalized to unity. At each time slot $t$, 
the network designer seeks to dynamic allocate resources to the agents, whose
types and utilities may also change over time. Let $\bs{x}_{i,t}$ be the resource allocated to agent $i$ and time $t$. 

The type of agent $i$ is $\theta_{i,t}\in\mathcal{T}$ at time $t$, which has an (instantaneous) utility $u(\theta_{i,t},\bs{x}_{i,t})$ at time $t$. Let $\rho_{\theta,t}$ be the portion of type-$\theta$ agents at time $t$.
For presentation simplicity,
we consider simple resource capacity constraints, given by
\begin{align}
    \sum_{i\in\mathcal{I}}x_{i,n,t}\leq C_n,~\forall n\in\mathcal{N},\label{constraint11-t}
\end{align}
That is, we assume that agents are only characterized by $\theta_{i,t}$ and we can drop the type coefficients $\zeta_{i}$ here.

Agents' types $\{\theta_{i,t}\}$ follow a general Markov process, where the Markov processes of agents are independent across agents.
\begin{align}
    P(\theta_{i,t+1}|\bs\theta_{i,1:t},\bs{z}_{\theta_i,1:t}) = Q(\theta_{i,t+1}|\theta_{i,t},\bs{z}_{\theta_i,t}).
\end{align}

All agents discount the future with a common discount factor $\delta\in(0,1)$. The network designer aims at maximizing the expected discounted network utility to be introduced later.

The dynamic NUM problem starting in period $t$ at state $\rho_t$ can be written as
\begin{align}
    W(\bs{\rho}_t)\triangleq& \max_{\bs{z}_t\in\mathcal{Z}}\left[ \sum_{\theta\in\mathcal{T}}\rho_{\theta,t}u(\theta_{t},\bs{z}_{\theta,t}) + \mathbb{E}\left[\sum_{s=t+1}^\infty \delta^{s-t}\right.\right.\nonumber\\
    &\sum_{j\in\mathcal{J}}\rho_{j,s}  u_i(z_{i,s},\theta_{j,s}) \Big| \rho_{t},z_t\Big]\Big]\nonumber\\
     =& \max_{\bs{z}_t\in\mathcal{Z}}\left[ \sum_{j\in\mathcal{J}}\rho_{j,t}u_i(z_{i,t},\theta_{j,t}) +\delta \mathbb{E}\left[W(\rho_{t+1}) | \rho_{t},z_t\right]\right],
\end{align}    
where $\mathcal{Z}$ is the network resource budget constraint such that
\begin{align}
    \sum_{\theta\in\mathcal{T}}\rho_{\theta,t}z_{\theta,n,t}\leq C_n,~\forall n\in\mathcal{N}.\label{constraint11-tt}
\end{align}

The (long-term discounted) network utility maximizing policy $\bs{z}_t^o$ can be hence given by:
\begin{align}
    &\bs{z}_t^o(\bs{\rho}_t)\nonumber\\
    =&\arg \max_{\bs{z}_t\in\mathcal{Z}}\left[ \sum_{\theta\in\mathcal{T}}\rho_{\theta,t}u(\theta_{t},\bs{z}_{\theta,t})+\delta \mathbb{E}\left[W(\bs\rho_{t+1}) | \bs{\rho}_{t},\bs{z}_{t}\right]\right].
    \end{align}
We then define a Markovian policy $\sigma$ for each agent. Its (long-term discounted) utility is given by $U^{\sigma}(\theta_{i,t},\bs{z}_{i,t},\rho_t)$ such that, for $\bs{z}_{i,t} = \sigma(\theta_{i,t},\rho_t)$,
    \begin{align}
    &U^{\sigma}(\theta_{i,t},\bs{z}_{i,t},\bs{\rho}_t)\nonumber\\
    &=u(\theta_{i,t},\bs{z}_{i,t})\nonumber\\
    &+\mathbb{E}\left[\sum_{s=t+1}^\infty \delta^{s-t}\sum_{\theta\in\mathcal{T}}\rho_{\theta,t}  u(\theta,\bs{z}_{\theta,s}^o(\rho_{s})) | \theta_{i,t},\bs{z}_{i,t},\bs\rho_t\right]
    \nonumber\\
    &=u(\theta_{i,t},\bs{z}_{i,t})\nonumber\\
    &+\delta \mathbb{E}\left[\sum_{\theta\in\mathcal{T}}\rho_{\theta,t+1}  U^{\sigma}(\theta_{t+1},\bs{z}_{\theta,t+1},\bs\rho_{t+1}) | \theta_{i,t},\bs{z}_{i,t},\bs\rho_t\right],
\end{align}
based on which, we can rewrite    $W(\rho_t)$ as
\begin{align}   
    W(\rho_t)&\triangleq \max_{\bs{z}_t\in\mathcal{Z}} \sum_{\theta\in\mathcal{T}}\rho_{\theta,t}u(\theta_{t},\bs{z}_{\theta,t}) +\delta \mathbb{E}\left[W(\bs\rho_{t+1}) | \bs\rho_{t},\bs{z}_t\right]\nonumber\\
    &=\max_{\bs{z}_t\in\mathcal{Z}} \sum_{\theta\in\mathcal{T}}\rho_{\theta,t}U^{\sigma}(\theta_{t},\bs{z}_{\theta,t},\rho_t).
\end{align}
That is, the long-term discounted network utility can be expressed as the sum of individual utilities.

\subsection{Dynamic Large-Scale VCG Mechanism}
Based on the above problem formulation in dynamic settings, we are ready to introduce the following dynamic version of the large-scale mechanism.

\begin{mechanism}[Dynamic Large-Scale VCG Mechanism]
The mechanism allocates resources and determines the payment for each agent $i$ in each time slot $t$. 
\begin{itemize}
    \item Message Space at time $t$: each agent $i$ reports its type $\tilde{\theta}_{i,t}\in\mathcal{T}$;
    \item Allocation Outcome at time $t$: let $\tilde{\rho}_{\theta,t}(\tilde{\bs{\theta}}_t)$ be the portion of reported type-$\theta$ agents for all $\theta\in\mathcal{T}$ at time $t$. The network designer solves the following problem:
\begin{subequations}\label{M2}
\begin{align}
&\max_{\bs{z}_t}~\sum_{\theta\in\mathcal{T}}~\tilde{\rho}_{\theta,t}(\tilde{\bs{\theta}}) U^{\sigma}_i( \theta_{j,t}, z_{\theta,t}, \rho_t)\\
&~~{\rm s.t.}~\sum_{\theta\in\mathcal{T}}\tilde{\rho}_{\theta,t}(\tilde{\bs{\theta}}) z_{\theta,t,n}\leq C_n,~\forall n\in\mathcal{N}\label{constraintaa}
\end{align}
\end{subequations}
Let $\bs{z}_t^*(\tilde{\bs{\theta}_t})$ be the optimal primal solution and $\bs{p}_t^*(\tilde{\bs{\theta}_t})$ be the optimal dual solution corresponding to the constraint in \eqref{constraintaa}. Each agent $i$ receives an allocation 
\begin{align}
    \bs{x}_{i,t}^*(\tilde{\bs{\theta}}_t)=\bs{z}_{\tilde{\theta}_i,t}^*(\tilde{\bs{\theta}}_t),
\end{align}
and needs to pay
\begin{align}
    h_{i,t}^*(\tilde{\bs{\theta}}_t)=\sum_{n\in\mathcal{N}}p_{t,n}^*(\tilde{\bs{\theta}}_t)z_{\tilde{\theta}_i,t,n}^*(\tilde{\bs{\theta}}_t).\label{Dy-pay}
\end{align}
\end{itemize}
\end{mechanism}

The dynamic mechanism yields a dynamic game: given all other agents' reports' $\bs{\theta}_{-i}$, each agent's long-term discounted payoff maximization problem is:
\begin{align}
\label{payoff}
   \max_{\tilde{\theta}_{i,t}\in\mathcal{T}}~P_{i,t}(\tilde{\theta}_{i,t},\tilde{\bs{\theta}}_{-i},{\bs{\rho}}_t)\triangleq
   &U^\sigma(\theta_{i,t},\bs{z}_{\tilde{\theta}_i,t}^*(\tilde{\theta}_{i,t},\tilde{\bs{\theta}}_{-i,t}),\bs\rho_t)\nonumber\\
   &-\sum_{n\in\mathcal{N}}p_{t,n}^*(\tilde{\bs{\theta}}_t)z_{\tilde{\theta}_i,t,n}^*(\tilde{\bs{\theta}}_t).
\end{align}
Note that $\bs{p}_t^*(\cdot)$ is determined in \eqref{M}. As $I\rightarrow \infty$, the value of 
$\bs\rho_\theta^*(\tilde{\bs{\theta}})$ remains the same even if agent $i$ solely deviates from reporting $\theta_{i,t}$, so does $\bs{p}_t^*(\tilde{\bs{\theta}})$. Hence, we can safely remove $\theta_{i,t}$ from $\bs{p}_t^*(\cdot)$. We are now ready to extend our results in Proposition \ref{P1} to the dynamic setting:

\begin{proposition}\label{P3}
When $I\rightarrow\infty$, the Dynamic Large-Scale VCG Mechanism satisfies (P1) and (P2).
\end{proposition}

To prove Proposition \ref{P3}, we note that each agent's payoff in \eqref{payoff} is bounded by
\begin{align}
P_{i,t}(\theta_{i,t},\tilde{\bs{\theta}}_{-i,t},\bs\rho_t)\leq \max_{z_{i,t}\geq 0} U_i^\sigma(z_{i,t},\theta_{i,t}, \rho_t) -p_t^*(\tilde{\bs{\theta}}_{-i,t})z_{i,t},\label{Eq}
\end{align}
which can be readily shown by contradiction.
In addition, the maximizer $\bs{z}_{i,t}^o(\tilde{\bs{\theta}}_{-i,t})$ of $$ U^\sigma(\theta_{i,t},\bs{z}_{\theta_i,t},\bs\rho_t)-\sum_{n\in\mathcal{N}}p_{t,n}^*(\tilde{\bs{\theta}}_t)z_{{\theta}_i,t,n}$$ satisfies
\begin{align}
    \frac{\partial U^\sigma(\theta_{i,t},\bs{z}_{i,t}, \bs\rho_t) }{\partial z_{i,t,n}}= p_{t,n}^*(\tilde{\bs{\theta}}_{-i,t}),~\forall i\in\mathcal{I}.\label{Eq2}
\end{align}
By analyzing the KKT conditions of \eqref{M}, $\bs{z}_{i,t}^o(\tilde{\bs{\theta}}_{-i,t})$ is exactly achieved when agent $i$ reports $\tilde{\theta}_{i,t}=\theta_{i,t}$.
From \eqref{payoff},
\begin{align}
    &~~P_{i,t}(\theta_{i,t},\tilde{\bs{\theta}}_{-i,t},\bs\rho_t)\nonumber\\
    &=U^\sigma(\theta_{i,t},\bs{z}_{i,t}^o(\tilde{\bs{\theta}}_{-i,t}),\bs\rho_t)-\sum_{n\in\mathcal{N}}p_{t,n}^*(\tilde{\bs{\theta}}_{-i,t})z_{i,t,n}^o(\tilde{\bs{\theta}}_{-i,t})\nonumber\\
&=\max_{\bs{z}_{i,t}\succeq\bs{0}} U^\sigma(\theta_{i,t},\bs{z}_{i,t},\bs\rho_t)-\sum_{n\in\mathcal{N}}p_{t,n}^*(\tilde{\bs{\theta}}_{-i,t})z_{i,t,n},~\forall i\in\mathcal{I}.\label{DPayoff}
\end{align}
From \eqref{Eq}, we see that the truthful report $\theta_{i,t}=\tilde{\theta}_{i,t}$ maximizes each agent $i$'s payoff. Therefore, the Dynamic Large-Scale VCG Mechanism achieves dominant-strategy incentive compatibility when $I\rightarrow\infty$.

\begin{proposition}
The Dynamic Large-Scale VCG Mechanism satisfies IR and weakly BB.
\end{proposition}

\begin{IEEEproof}
The dual variables $\bs{p}_t^*(\tilde{\bs{\theta}})$ are component-wisely positive. Therefore, each agent $i$'s payment
$h_{i,t}^*(\tilde{\bs{\theta}_t})$ in \eqref{Dy-pay} is also positive. This  shows that weakly BB is satisfied. 

From \eqref{DPayoff}, each agent's payoff when truthful reporting satisfies 
\begin{align}
    P_{i,t}(\theta_{i,t},\tilde{\bs{\theta}}_{-i,t})&=\max_{\bs{z}_{i,t}\succeq \bs{0}} U^\sigma(\theta_{i,t},\bs{z}_{i,t},\bs\rho_t)-\sum_{n\in\mathcal{N}}p_{t,n}^*(\tilde{\bs{\theta}}_{-i,t})z_{i,t,n}\nonumber\\
    &\geq U^\sigma(\theta_{i,t},\bs{0},\bs\rho_t)-\sum_{n\in\mathcal{N}}p_{t,n}^*(\tilde{\bs{\theta}}_{-i,t})\cdot 0\geq 0,
\end{align}
which shows that IR is achieved.
\end{IEEEproof}

\begin{theorem}\label{T3}
Given a finite number of agents $I$, the Dynamic Large-Scale VCG Mechanism is
$\epsilon$-incentive compatible for all $\epsilon$ satisfying
\begin{align}
    \epsilon\geq \frac{2}{I^2}\frac{\sum_{\theta\in\mathcal{T}}L_{\theta}\sum_{n\in\mathcal{N}}C_n^2}{ L_f^2 \min_{r,t}\rho_{r,t}^4}.
\end{align}
\end{theorem}
Theorem \ref{T3} can be proved in a similar way as Theorem \ref{T1}. The proof is lengthy and hence is omitted.

\section{Conclusion}
\label{sec:Concl}
In this paper, we have considered mechanism design for large-scale networked systems. We have presented the Large-Scale VCG Mechanism which maximizes the network utility and
achieves individual rationality and budget balance. We have
proven that each agent's incentive to misreport its type converges quadratically to zero. We have also extended the results to a dynamic setting.

\appendices
    \section{Proof Sketch of Lemma \ref{L1}}\label{ProofL1}

Define
\begin{align}
U_{-i}(\bs{x}_{-i})\triangleq \sum_{j\neq i}U_j(\tilde{\theta}_j,\bs{x}_j)\nonumber\\
    \tilde{\bs{x}}_{-i}=\arg\max_{\bs{x}_{-i}\in\mathcal{X}_{-i}}U_{-i}(\bs{x}_{-i}).
\end{align}
The VCG payment for user $i$ is given by:
\begin{align}
h_i(r_i,\bs{r}_{-i})&=-\sum_{j\neq i}U(\tilde{\theta}_j,\bs{x}_j^*(\bs{r}))+\max_{\bs{x}_i\in\mathcal{X}_{-i}}\sum_{j\neq i}U(\tilde{\theta}_j,\bs{x}_j)\nonumber\\
&=\int_{\bs{x}_{-i}^*}^{\tilde{\bs{x}}_{-i}} \nabla_{\bs{x}_{-i}} U_{-i}(\bs{x}_{-i})d{\bs{x}_{-i}}\nonumber\\
&=\int_{0}^{1} \nabla_{\bs{x}_{-i}} U_{-i}((1-t)\bs{x}_{-i}^*+t\tilde{\bs{x}}_{-i})\nonumber\\
&~~\times[(1-t)\bs{x}_{-i}^*+t\tilde{\bs{x}}_{-i}]dt.\label{Eq2-1}
\end{align}

Note that, as $I\rightarrow \infty$, the changes in the equilibrium 
\begin{align}
 \nabla_{\bs{x}_i}U_{-i}((1-t)\bs{x}_{-i}^*+t\tilde{\bs{x}}_{-i})= \nabla_{\bs{x}_i}U_{-i}(\bs{x}_{-i}^*),~\forall t\in[0,1].\label{Eq2-2}
\end{align}
Since the NUM Problem in \eqref{NUM} is convex and the Slater's condition is satisfies, the sufficient and necessary KKT conditions of optimality yields
\begin{align}
    \frac{\partial U(\theta_i,\bs{x}_i)}{\partial x_{i,n}}=\lambda_n^o  \frac{\partial f_{\zeta_i,n}(x_{i,n}^o)}{\partial x_{i,n}},~\forall i\in\mathcal{I}, n\in\mathcal{N}.\label{KKTT}
\end{align}
In addition, $\sum_{j\neq i}(\tilde{x}_{j,n}-{x}^*_{j,n})=x_{i,n}^o$.
Combining \eqref{Eq2-1}-\eqref{KKTT}, we have 
\begin{align}
    h_i(\bs{r}^*) = \sum_{n\in\mathcal{N}}\lambda_n^o f_{\zeta_i,n}(x_{i,n}^o).
\end{align}

\section{Proof of Theorem \ref{T1}}\label{ProofT1}

Define a vector function 
\begin{align}
F(\bs{z},\bs{p},\bs{\rho})=\begin{pmatrix}
\nabla_{\bs{z}_{\theta,\zeta}} U(\theta,\bs{z}_{\theta,\zeta})-\bs{p}^T \bs{f}_{\zeta}(\bs{z}_{\theta,\zeta}),~t\in\mathcal{T}~\\
\sum_{\theta\in\Theta} \sum_{\zeta\in\mathcal{Z}} \rho_{\theta,\zeta} f_{\zeta,n}(z_{\theta,\zeta,n})- C_n,~n\in\mathcal{N}
\end{pmatrix}
\end{align}
The optimality condition of \eqref{Eq1} is $F=\bs{0}$. Let ${\rm diag}$ be the operator that places a set of scalars (or square matrices) into the diagonal of a square matrix to create a diagonal matrix (or block-diagonal matrix).
The function $F$ is continuously differentiable and
\begin{align}\label{Eq17}
    -\left(\frac{\partial F}{\partial \bs{z}}~\frac{\partial F}{\partial \bs{p}}\right)=\begin{pmatrix}
M_1 & M_2^T \\
-M_2 & \bs{0}
\end{pmatrix},
\end{align}
where 
\begin{align}
    M_1=-{\rm diag}(\nabla_{\bs{z}_r}^2 U_r(\bs{z}_r)),
\end{align}
\begin{align}
    M_2=\begin{pmatrix}D_1 & D_2 & ... & D_{|\mathcal{R}|}\end{pmatrix},
\end{align}
and 
\begin{align}
    D_r={\rm diag}(\rho_r f'_{r,n}(z_n)),~\forall r\in\mathcal{R}.
\end{align}
We note that \eqref{Eq17} is nonsinglar, and its inverse is
\begin{align}
    -\left(\frac{\partial F}{\partial \bs{z}}~\frac{\partial F}{\partial \bs{p}}\right)^{-1}=\begin{pmatrix}
H & -D^{-1}\bs{1}_{|\mathcal{T}|\times 1} Q \\
Q\bs{1}_{1\times |\mathcal{T}|}D^{-1} & Q
\end{pmatrix},
\end{align}
where $Q=(M_2M_1^{-1}M_2^T)^{-1}$ and $H=D^{-1}-D^{-1} \bs{1}_{|\mathcal{T}|\times 1} Q \bs{1}_{1 \times |\mathcal{T}|} D^{-1}$. Therefore, by the implicit function theorem, the functions $\bs{z}^*(\bs{\rho})$ and $p^*(\bs{\rho})$ are continuously differentiable, and \begin{align}
    \begin{pmatrix}\frac{\partial \bs{z}^*}{\partial \bs{\rho}}\\ \frac{\partial p^*}{\partial \bs{\rho}}\end{pmatrix}=-\left(\frac{\partial F}{\partial \bs{z}}~\frac{\partial F}{\partial \bs{p}}\right)^{-1}\cdot \frac{\partial F}{\partial \bs{\rho}}.
\end{align}
Hence, we have 
\begin{align}
    \frac{\partial \bs{p}^*}{\partial \bs{\rho}}=Q \bs{f},
\end{align}
where $\bs{f}=\{\bs{f}_{r}(\bs{z}_r^*)\}_{r\in\mathcal{R}}$ and $\bs{f}_{r}(\bs{z}_r^*)=\{{f}_{\tilde{\zeta},n}({z}_{r,n}^*)\}_{n\in\mathcal{N}}$.
That is,  for all $r\in\mathcal{R}$,
\begin{align}
    \frac{\partial \bs{p}^*}{\partial {\rho}_r}=&
    \left(\sum_{r'\in\mathcal{R}}D_{r'}(\nabla^2_{\bs{z}_{r'}}U_{r'}(\bs{z}_{r'}))^{-1}D_{r'}\right)^{-1}\bs{f}_r(\bs{z}_{r}).
    \label{z1}
\end{align}
Hence, we have
\begin{align}
   &\Big|\bs{f}_r(\bs{z}_r)^T \frac{\partial \bs{p}^*}{\partial {\rho}_r}\Big|\nonumber\\
   =~&\Big|\bs{f}_r(\bs{z}_{r})^T\left(\sum_{r'\in\mathcal{R}}D_{r'}(\nabla^2_{\bs{z}_{r'}}U_{r'}(\bs{z}_{r'}))^{-1}D_{r'}\right)^{-1}\bs{f}_r(\bs{z}_{r})\Big|,\nonumber\\
   \overset{(a)}{\leq}~& \Big|\bs{f}_r(\bs{z}_{r})^T\frac{\sum_{r'\in\mathcal{R}}D^{-1}_{r'}\nabla^2_{\bs{z}_{r'}}U_{r'}(\bs{z}_{r'})D^{-1}_{r'}}{I}\bs{f}_r(\bs{z}_{r})\Big|,\nonumber\\
   \leq ~&\Big|\bs{f}_r(\bs{z}_{r})^T\frac{\sum_{r'\in\mathcal{R}}\nabla^2_{\bs{z}_{r'}}U_{r'}(\bs{z}_{r'})}{ I \cdot L_f^2 \min_r\rho_r^2}\bs{f}_r(\bs{z}_{r})\Big|,\nonumber\\
   \leq ~&\frac{|\mathcal{Z}|\sum_{\theta\in\mathcal{T}}L_{\theta}}{ I \cdot L_f^2 \min_r\rho_r^2}||\bs{f}_r(\bs{z}_{r})||_2^2\leq \frac{|\mathcal{Z}|\sum_{\theta\in\mathcal{T}}L_{\theta}\sum_{n\in\mathcal{N}}C_n^2}{ I \cdot L_f^2 \min_r\rho_r^4}.
\end{align}
where $(a)$ is the due to the harmonic-mean-arithmetic-mean matrix inequality \cite{matrix}.
From \eqref{z1}, if a type-$r$ user reports its type as $j$, the change in its payment,  $\Delta h_{r\rightarrow j}$, is bounded by 
\begin{align}
    \Delta h_{r\rightarrow j}&\leq \max\left|\bs{f}_r(\bs{z}_r)^T \frac{\partial \bs{p}^*}{\partial {\rho}_r}\right| |\Delta \rho_r|+\max \left|\bs{f}_j(\bs{z}_j)^T \frac{\partial \bs{p}^*}{\partial {\rho}_j} \right| |\Delta\rho_j|,\nonumber\\
    &= \frac{2}{I}\max\left|\bs{f}_r(\bs{z}_r)^T \frac{\partial \bs{p}^*}{\partial {\rho}_r}\right|,\nonumber\\
    &\leq \frac{2}{I^2}\frac{|\mathcal{Z}|\sum_{\theta\in\mathcal{T}}L_{\theta}\sum_{n\in\mathcal{N}}C_n^2}{ L_f^2 \min_r\rho_r^4}.\label{hahaha}
\end{align}

Finally, each type-$r_i=(\theta_i,\zeta_i)$ agent's payoff satisfies that, for all other report $\hat{r}_i=(\hat{\theta}_i,\hat{\zeta}_i)\in\mathcal{R}$, we have
\begin{align}
    &U(\theta_i,\bs{z}_{r_i}^*(r_i,\bs{r}_{-i}))-
    h_i(r_i,\bs{r}_{-i})+\frac{2}{I^2}\frac{|\mathcal{Z}|\sum_{\theta\in\mathcal{T}}L_{\theta}\sum_{n\in\mathcal{N}}C_n^2}{  L_f^2 \min_r\rho_r^4}\nonumber\\
    \overset{(a)}{\geq} &U(\hat{\theta}_i,\bs{z}_{\hat{r}_i}^*(\hat{r}_i,\bs{r}_{-i}))-
    h_i(
    \hat{r}_i,\bs{r}_{-i})+\frac{2}{I^2}\frac{|\mathcal{Z}|\sum_{\theta\in\mathcal{T}}L_{\theta}\sum_{n\in\mathcal{N}}C_n^2}{ L_f^2 \min_r\rho_r^4}\nonumber\\
    \overset{(b)}{\geq} & U(\hat{\theta}_i,\bs{z}_{\hat{r}_i}^*(\hat{\theta}_i,\bs{r}_{-i}))-
    (h_i(\hat{\theta}_i,\bs{r}_{-i})-\Delta h_{{r_i}\rightarrow \hat{r_i}})\nonumber\\
    \geq & P_i(\hat{r}_i,\bs{r}_{-i}),
\end{align}
where (a) is due to Proposition \ref{P1}; (b) comes from \eqref{hahaha}. We complete the proof.

\end{document}